\documentclass[letterpaper,10pt,conference]{ieeeconf}

\usepackage{cite}
\usepackage{graphicx}
\usepackage{tikz}
\usepackage{subcaption}
\usepackage[cmex10]{amsmath}
\usepackage{algpseudocode}
\usepackage{algorithmicx} 
\usepackage{array}

\usepackage{epsfig}
\usepackage{amsfonts,amssymb,multirow,bigstrut,booktabs,ctable,latexsym}

\usepackage[mathscr]{eucal}
\usepackage{verbatim}

\usepackage{amsthm}
\usepackage{algorithm}

\usepackage{stmaryrd}

\IEEEoverridecommandlockouts  
\begin{document}

\newtheorem{definition}{Definition}
\newtheorem{lemma}{Lemma}
\newtheorem{corollary}{Corollary}
\newtheorem{theorem}{Theorem}
\newtheorem{example}{Example}
\newtheorem{proposition}{Proposition}
\newtheorem{remark}{Remark}
\newtheorem{assumption}{Assumption}
\newtheorem{corrolary}{Corrolary}
\newtheorem{property}{Property}
\newtheorem{ex}{EX}
\newtheorem{problem}{Problem}
\newcommand{\argmin}{\arg\!\min}
\newcommand{\argmax}{\arg\!\max}
\newcommand{\st}{\text{s.t.}}
\newcommand \dd[1]  { \,\textrm d{#1}  }

\makeatletter
\newsavebox{\@brx}
\newcommand{\llangle}[1][]{\savebox{\@brx}{\(\m@th{#1\langle}\)}%
  \mathopen{\copy\@brx\kern-0.5\wd\@brx\usebox{\@brx}}}
\newcommand{\rrangle}[1][]{\savebox{\@brx}{\(\m@th{#1\rangle}\)}%
  \mathclose{\copy\@brx\kern-0.5\wd\@brx\usebox{\@brx}}}
\makeatother

\title{\Large\bf Abstraction-Free Control Synthesis to Satisfy Temporal Logic Constraints under Sensor Faults and Attacks}

\author{Luyao Niu$^1$, Zhouchi Li$^2$, and Andrew Clark$^3$ %
\thanks{$^1$Luyao Niu is with the Network Security Lab, Department of Electrical and Computer Engineering,
University of Washington, Seattle, WA 98195-2500,
        {\tt\small luyaoniu@uw.edu}}
\thanks{$^2$Z. Li is with the Department of Electrical and Computer Engineering, Worcester Polytechnic Institute, Worcester, MA 01609,
{\tt\small zli4@wpi.edu}}
\thanks{$^3$Andrew Clark is with the Electrical and Systems Engineering Department, McKelvey School of Engineering, Washington
University in St. Louis, St. Louis, MO 63130,
        {\tt\small andrewclark@wustl.edu}}%
\thanks{This work was supported by the National Science Foundation and the Office of Naval Research via grants CNS-1941670 and N00014-17-S-B001.}
}
\thispagestyle{empty}
\pagestyle{empty}

\maketitle

\begin{abstract}
We study the problem of synthesizing a controller to satisfy a complex task in the presence of sensor faults and attacks. We model the task using Gaussian distribution temporal logic (GDTL), and propose a solution approach that does not rely on computing any finite abstraction to model the system. We decompose the GDTL specification into a sequence of reach-avoid sub-tasks. We develop a class of fault-tolerant finite time convergence control barrier functions (CBFs) to guarantee that a dynamical system reaches a set within finite time almost surely in the presence of malicious attacks. We use the fault-tolerant finite time convergence CBFs to guarantee the satisfaction of `reach' property. We ensure `avoid' part in each sub-task using fault-tolerant zeroing CBFs. These fault-tolerant CBFs formulate a set of linear constraints on the control input for each sub-task. We prove that if the error incurred by system state estimation is bounded by a certain threshold, then our synthesized controller fulfills each reach-avoid sub-task almost surely for any possible sensor fault and attack, and thus the GDTL specification is satisfied with probability one. We demonstrate our proposed approach using a numerical study on the coordination of two wheeled mobile robots.

\end{abstract}

\section{Introduction}\label{sec:intro}

Cyber-physical systems (CPS) in applications including autonomous vehicles, industrial control systems, and robotics must complete increasingly complex tasks. Temporal logics \cite{baier2008principles} such as linear temporal logic (LTL) are used to specify tasks because of their concreteness, rigor, and rich expressiveness. To model the process and observation noises that are commonly incurred by CPS, Gaussian distribution temporal logic (GDTL) has been proposed to specify CPS properties involving the uncertainties in the system state \cite{vasile2016control,leahy2019control}. Controller design for CPS under different temporal logic constraints has attracted extensive research attention \cite{kress2009temporal,coogan2015traffic,vasile2016control,leahy2019control}. There have been two types of methodologies for temporal logic based control synthesis, namely \emph{abstraction-based} and \emph{abstraction-free} approaches. 

The abstraction-based approaches \cite{fainekos2009temporal,ding2014optimal,wongpiromsarn2009receding,vasile2016control,leahy2019control,kress2009temporal,coogan2015traffic} lift the CPS model from continuous domain to discrete domain by computing a finite abstraction, e.g., finite transition systems \cite{fainekos2009temporal} and Markov decision processes (MDPs) \cite{ding2014optimal,wongpiromsarn2009receding}. 
A controller is synthesized by applying model checking algorithms on the finite abstraction. In general, the abstraction-based approaches are computationally demanding and suffer from the curse of dimensionality. To mitigate the computation required by the abstraction-based approaches, abstraction-free approaches have been proposed \cite{wolff2012robust,horowitz2014compositional,papusha2016automata,srinivasan2018control,lindemann2018control,yang2019continuous,niu2020control}. These approaches eliminate the step of constructing the finite abstractions, and directly compute the controller in the continuous state space.

The previously introduced abstraction-based \cite{fainekos2009temporal,ding2014optimal,wongpiromsarn2009receding,vasile2016control,leahy2019control,kress2009temporal,coogan2015traffic} and abstraction-free approaches \cite{wolff2012robust,horowitz2014compositional,papusha2016automata,srinivasan2018control,lindemann2018control,yang2019continuous,niu2020control} assume that there exist no faults or malicious attacks targeting the system sensors. However, CPS have been shown to be vulnerable to sensor faults and malicious attacks \cite{liu2011false,mo2010false,cardenas2008research,fawzi2014secure}, making temporal logic based control synthesis more challenging. Sensor faults and malicious attacks can arbitrarily manipulate the sensor measurements. Inaccurate measurements prevent the CPS from correctly evaluating and tracking the satisfaction of temporal logic specifications. Moreover, the faulty or compromised measurements can bias the state estimate and lead to erroneous controller actions, which deviate the system trajectory and violate the temporal logic specification. Although there has been research attention on fault and attack detection \cite{mo2010false,shoukry2017secure,fawzi2014secure} and low-level control design under sensor attacks \cite{clark2018linear}, abstraction-free temporal logic based control synthesis under sensor faults and attacks has not been studied.

In this paper, we study the problem of synthesizing a controller for a nonlinear control-affine system to satisfy a GDTL specification in the presence of process and observation noises as well as malicious attacks on sensor measurements. We assume that the adversary has finitely many choices of the sensors to attack. The system estimates the state by employing a set of Extended Kalman Filters (EKFs) associated with each possible set of sensors attacked by the adversary. 
If there exists a control input that satisfies the GDTL specification for all of the state estimates, then the GDTL formula can be satisfied regardless of the attack. When such control input does not exist, it indicates that the state estimates given by some EKFs conflict with the others. We view two EKFs as conflicting if the divergence between their state estimates exceeds a certain threshold. We then compute a baseline state estimate without using any sensor corresponding to the conflicting EKFs. We compare this baseline with those from the conflicting EKFs, and treat the one that diverges from the baseline as being compromised.
We make the following specific contributions:
\begin{itemize}
    \item We decompose the GDTL specification into a sequence of reach-avoid sub-tasks with the system state uncertainties being explicitly encoded.
    \item We synthesize the controller for each decomposed problem by deriving a set of sufficient conditions. We guarantee the `reach' and `avoid' properties in each sub-task by developing a class of fault-tolerant finite time convergence control barrier functions and applying the fault-tolerant zeroing control barrier functions.
    \item We prove that if the errors incurred by the EKFs are bounded by some thresholds, then each decomposed reach-avoid sub-task is satisfied with probability one using our proposed approach, and thus the GDTL specification is satisfied almost surely.
    \item We evaluate our proposed approach using a numerical case study on two wheeled mobile robots. We show that our proposed approach satisfies the given GDTL specification, while a baseline fails.
\end{itemize}

The rest of this paper is organized as follows. Section \ref{sec:related} reviews related work. Section \ref{sec:prelim} presents the system model and reviews necessary background on finite time stability in probability, EKF, and GDTL. Section \ref{sec:sol} presents our proposed solution approach. A numerical case study is presented in Section \ref{sec:simulation}. Section \ref{sec:conclusion} concludes the paper.

\section{Related Work}\label{sec:related}

Abstraction-based approaches have been proposed for control synthesis of CPS in the absence of observation noise and sensor attacks under LTL constraints \cite{kress2009temporal,fainekos2009temporal,ding2014optimal,wongpiromsarn2009receding}. When CPS are subject to observation noises, GDTL has been proposed to model the uncertainties \cite{vasile2016control,leahy2019control}. 

When CPS operate under malicious attacks targeting the system actuators, abstraction-based approaches for control synthesis under LTL constraints have been studied in \cite{niu2019optimal,niu2020optimal}, where the CPS are abstracted as finite stochastic games. For CPS under partially observable environments, partially observable stochastic game is used as the finite abstraction for control synthesis \cite{ramasubramanian2020secure}. These works \cite{niu2019optimal,niu2020optimal,ramasubramanian2020secure} and the present paper focus on different scopes of attacks. In \cite{niu2019optimal,niu2020optimal,ramasubramanian2020secure}, sensor faults or attacks are not considered, and thus the system state is known by the CPS. This paper focuses on sensor faults and attacks. In addition, our proposed approach does not compute a finite abstraction for the CPS, and thus belongs to the class of abstraction-free approaches.

Abstraction-free approaches eliminate the step for computing the finite abstractions to improve the scalability. In \cite{wolff2014optimization}, the LTL specification is encoded as a set of mixed integer linear constraints. In \cite{horowitz2014compositional}, the authors solve a sequence of stochastic reachability problems via nonlinear PDEs to satisfy a co-safe LTL specification. Alternatively, a mixed continuous-discrete HJB-based formulation is proposed in \cite{papusha2016automata}. In general, solving PDEs and HJB equations are computationally expensive.

In \cite{lindemann2018control,srinivasan2018control,niu2020control}, the authors adopt control barrier functions (CBFs), which are originally proposed for safety-critical control synthesis \cite{wang2017safety,ames2016control}, to compute control policies that satisfy LTL specifications for deterministic CPS. In parallel, (control) barrier certificate based verification and synthesis for temporal logic properties have been investigated in \cite{jagtap2020formal,anand2019verification}. An abstraction-free safe learning scheme is proposed in \cite{sun2020continuous} for systems under adversarial actuator attacks. Our paper differs from these barrier function based works in two aspects. First, we consider the presence of sensor faults and attacks. Second, we develop a class of fault-tolerant finite time convergence CBF, in combined with the recent advancement in fault-tolerant zeroing CBF, to address the presence of noises, faults, and attacks.

\section{Preliminaries and Problem Formulation}\label{sec:prelim}

In this section, we first present the system model. We then give some background on GDTL and EKF. We finally present the problem formulation.

\subsection{System Model}\label{sec:system}

In this subsection, we first give the notations that we will use throughout this paper. We next present the system model.

\textbf{Notations.} We use $\mathbb{R}^n$ to denote the $n$-dimensional Euclidean space. Given a vector $a\in\mathbb{R}^n$, we use $supp(a)$ and $[a]_i$ to denote its support and $i$-th entry, respectively. For any matrix $A$, we denote its trace as $tr(A)$. The set of positive semi-definite matrices of dimension $n\times n$ is represented as $\mathbb{S}^{n\times n}$. We denote the set of non-negative real numbers as $\mathbb{R}_{\geq 0}$. 
A function $\alpha:\mathbb{R}_{\geq 0}\rightarrow\mathbb{R}_{\geq 0}$ belongs to class $\mathcal{K}$ if it is continuous, strictly increasing, and $\alpha(0)=0$. A class $\mathcal{K}$ function belongs to class $\mathcal{K}_\infty$ if $\alpha(x)\rightarrow \infty$ when $x\rightarrow \infty$. We use $\mathbb{P}(\cdot)$ to represent the probability of an event.

Consider a nonlinear dynamical system defined as
\begin{equation}\label{eq:state dynamics}
    dx_t = (f(x_t)+g(x_t)u_t)\ dt+\sigma_t\ dW_t,
\end{equation}
where $x_t\in\mathbb{R}^n$ is the system state, $u_t\in\mathbb{R}^p$ is the control input at time $t$, $\sigma_t\in\mathbb{R}^{n\times n}$, and $W_t$ is an $n$-dimensional Brownian motion. Functions $f:\mathbb{R}^n\rightarrow\mathbb{R}^n$ and $g:\mathbb{R}^n\rightarrow\mathbb{R}^{n\times p}$ are locally Lipschitz continuous.

We denote the system output at time $t$ as $y_t\in\mathbb{R}^q$. The output $y_t$ follows an observation process \cite{fujisaki1972stochastic} given as:
\begin{equation}\label{eq:output}
    dy_t = (cx_t+a_t)\ dt+\nu_t\ dV_t,
\end{equation}
where $c\in\mathbb{R}^{q\times n}$, $\nu_t\in\mathbb{R}^{q\times q}$, and $V_t$ is a $q$-dimensional Brownian motion.
We assume that the system operates in the presence of a malicious adversary. The adversary can inject a signal $a_t\in\mathbb{R}^q$ to manipulate the output $y_t$ with $supp(a_t)\subseteq\{1,\ldots,q\}$. We assume that once the adversary determines the set of sensors to attack, it cannot adjust the support of the attack signal, i.e., the adversary cannot change its target after the attack starts. We assume that there are finitely many choices for the support of $a_t$, denoted as $\mathcal{F}=\{r_1,\ldots,r_m\}$, which is also known to the system. In the remainder of this paper, we refer to $r_i$ as a fault pattern. However, the choice $supp(a_t)$ made by the adversary is not known to the system.

\subsection{Finite Time Stability in Probability of Stochastic Nonlinear Systems}


In this subsection, we introduce background on finite time stability of stochastic nonlinear systems. Consider a stochastic nonlinear system with complete state information given by Eqn. \eqref{eq:state dynamics}. We then give the definition of finite time stability in probability.
\begin{definition}[Finite Time Stability in Probability \cite{yin2011finite}]\label{def:finite time stability}
The solution of system \eqref{eq:state dynamics}, denoted as $x(t;x_0)$, is said to be finite time stable in probability if system \eqref{eq:state dynamics} admits a unique solution for any initial state $x_0\in\mathbb{R}^n$ and the following conditions hold:
\begin{enumerate}
    \item \underline{\emph{Finite time attractiveness in probability:}} For every initial state $x_0\neq \mathbf{0}$, the stochastic settling time $\tau_{x_0}=\inf\{t:x(t;x_0)=\mathbf{0}\}$ is finite almost surely, i.e., $\mathbb{P}(\tau_{x_0}<\infty)=1$.
    \item \underline{\emph{Stability in probability:}} For every pair of $\epsilon\in(0,1)$ and $\chi>0$, there exists some constant $v$ such that 
    \begin{equation}
        \mathbb{P}(|x(t;x_0)|<\chi,~\forall t\geq 0)\geq 1-\epsilon,
    \end{equation}
    for any $|x_0|<v$.
\end{enumerate}
\end{definition}

Let $V$ be a twice continuously differentiable function. Finite time stability in probability is then certified by the following stochastic Lyapunov theorem.
\begin{theorem}[Stochastic Lyapunov Theorem \cite{yin2011finite}]\label{thm:SCLF}
Assume system \eqref{eq:state dynamics} admits a unique solution. If there exists a twice continuously differentiable function $V:\mathbb{R}^n\rightarrow\mathbb{R}_{\geq 0}$, class $\mathcal{K}_\infty$ functions $\alpha_1$ and $\alpha_2$, positive constants $\gamma>0$, and $\rho\in(0,1)$ such that 
\begin{subequations}
\begin{align}
    &\alpha_1(|x|)\leq V(x)\leq \alpha_2(|x|),\\
    &\frac{\partial V}{\partial x}(x)(f(x)+g(x)u)+\frac{1}{2}tr\left(\sigma^\top\frac{\partial^2 V(x)}{\partial x^2}\sigma\right)\nonumber\\
    &\quad\quad\quad\quad\quad+ \gamma\cdot V(x)^\rho\leq 0
\end{align}
\end{subequations}
then the solution of system \eqref{eq:state dynamics} is finite time stable in probability.
\end{theorem}

\subsection{Extended Kalman Filter}\label{sec:EKF}

The system given in Eqn. \eqref{eq:state dynamics} and \eqref{eq:output} employs a set of EKFs \cite{reif2000stochastic} to estimate the system state. Each EKF corresponds to a potential fault pattern $r_i\in\mathcal{F}$. We briefly review the background on EKF in the following.

We denote the state estimate at time $t$ as $\hat{x}_{t,i}$ when fault pattern $r_i$ is chosen by the adversary. Then the dynamics of $\hat{x}_{t,i}$ can be represented as
\begin{equation}\label{eq:EKF estimate}
    \hat{x}_{t,i}=(f(\hat{x}_{t,i})+g(\hat{x}_{t,i})u_t)dt+K_{t,i}(dy_{t,i}-c_i\hat{x}_{t,i}dt),
\end{equation}
where $K_{t,i}=P_{t,i}c_i^\top R_{t,i}^{-1}$, $R_{t,i}=\nu_{t,i}\nu_{t,i}^\top$, and $c_i$ as well as $y_{t,i}$ are obtained from matrix $c$ and vector $y_t$ by removing the corresponding rows indexed by $r_i$. Matrix $\nu_{t,i}$ is obtained from $\nu_t$ by removing the rows and columns indexed by $r_i$. 
The covariance of $x_{t,i}$, denoted as $P_{t,i}$, is the solution to
\begin{equation}\label{eq:EKF Ricatti}
    \frac{dP_{t,i}}{dt}=A_{t,i}P_{t,i}+P_{t,i}A_{t,i}^\top +Q_{t}-P_{t,i}c_i^\top R_{t,i}^{-1}cP_{t,i},
\end{equation}
where $Q_t=\sigma_t\sigma_t^\top$ and $A_{t,i}=\frac{\partial \Bar{f}}{\partial x}(\hat{x}_t,u)$ with $\Bar{f}(\hat{x}_{t,i},u)=f(\hat{x}_{t,i})+g(\hat{x}_{t,i})u$. A belief state $b_{t,i}\in\mathbb{R}^{n}\times\mathbb{S}^{n\times n}\triangleq \mathbb{B}$ corresponding to fault pattern $r_i$ is defined as $b_{t,i}=(\hat{x}_{t,i},P_{t,i})$. The evolution of the belief state is described by Eqn. \eqref{eq:EKF estimate} and \eqref{eq:EKF Ricatti}. We make the following assumptions on Eqn. \eqref{eq:state dynamics} and \eqref{eq:output}.
\begin{assumption}\label{assump:SDE conditions}
SDEs in Eqn. \eqref{eq:state dynamics} and \eqref{eq:output} satisfy the following conditions for all fault patterns $r_i$:
\begin{enumerate}
    \item The pair $[\frac{\partial \Bar{f}}{\partial x}(\hat{x}_{t,i},u),c]$ is uniformly detectable;
    \item There exist constants $\beta_1$ and $\beta_2$ such that $Q_t\geq \beta_1I$ and $R_{t,i}\geq \beta_2I$ for all $t$;
    
    \item Let $$\Delta(x_t,\hat{x}_{t,i},u)=\Bar{f}(x_t,u)-\Bar{f}(\hat{x}_{t,i},u)-\frac{\Bar{f}}{\partial x}(x_t-\hat{x}_{t,i},u).$$ Then there exist real numbers $k_\Delta$ and $\epsilon_\Delta$ such that $\|\Delta(x_t,\hat{x}_{t,i},u)\|\leq k_\Delta \|x_t-\hat{x}_{t,i}\|_2^2$
    for all $x_t$ and $\hat{x}_{t,i}$ satisfying $\|x_t-\hat{x}_{t,i}\|_2\leq \epsilon_\Delta$.
\end{enumerate}
\end{assumption}
Given Assumption \ref{assump:SDE conditions}, we have the following result.
\begin{lemma}[\cite{reif2000stochastic}]\label{lemma:estimation error}
Suppose the conditions of Assumption \ref{assump:SDE conditions} hold. If there exists $\varepsilon>0$ such that if $Q_t\leq \varepsilon I$ and $R_{t,i}\leq \varepsilon I$, then for any $\zeta>0$, there exists $\gamma>0$ such that
\begin{equation}\label{eq:EKF accuracy}
    \mathbb{P}\left(\sup_{t\geq 0}\|x_t-\hat{x}_{t,i}\|_2\leq \gamma\right)\geq 1-\zeta.
\end{equation}
when the fault pattern is $r_i$.
\end{lemma}

\subsection{Gaussian Distribution Temporal Logic and Problem Formulation}

In this subsection, we introduce GDTL \cite{vasile2016control,leahy2019control}, which will be used to model the task that needs to be satisfied by system \eqref{eq:state dynamics}-\eqref{eq:output}. We then formulate the problem studied in this paper.

The syntax of GDTL is defined as follows:
\begin{equation}
    \varphi = True|l\leq 0|\neg\varphi|\varphi_1\land\varphi_2|\varphi_1\mathcal{U}\varphi_2,
\end{equation}
where $l:\mathbb{B}\rightarrow\mathbb{R}$, and $l\leq 0$ is a predicate.

Let $\mathbf{b}=b^0b^1\ldots$ be an infinite sequence of belief states with each $b^i\in\mathbb{B}$. We denote the suffix $b^tb^{t+1}\ldots$ of $\mathbf{b}$ as $\mathbf{b}^t$. Then a GDTL formula is interpreted as follows:
\begin{itemize}
    \item $\mathbf{b}^t\models True$
    \item  $\mathbf{b}^t\models l\leq 0$ iff $l(b^t)\leq 0$
    \item $\mathbf{b}^t\models\neg\varphi$ iff $\neg(\mathbf{b}^t\models\varphi)$
    \item $\mathbf{b}^t\models\varphi_1\land\varphi_2$ iff $(\mathbf{b}^t\models\varphi_1)\land (\mathbf{b}^t\models\varphi_2)$
    \item $\mathbf{b}^t\models\varphi_1\mathcal{U}\varphi_2$ iff $\exists t'\geq t$ s.t. $(\mathbf{b}^{t'}\models\varphi_2)\land (\mathbf{b}^{t''}\models\varphi_1)$ for all $t''<t'$
\end{itemize}

Let $\mathbb{B}_l=\{b:l(b)\leq 0\}$ be the set of belief states such that predicate $l$ is true and $\mathcal{L}_\varphi$ be the set of predicates that appear in a GDTL formula $\varphi$. We present the following proposition \cite{leahy2019control} so as to construct an equivalent deterministic Rabin automaton (DRA) representing GDTL formula $\varphi$.
\begin{proposition}[\cite{leahy2019control}]\label{def:LTL equiv}
Let $\varphi$ be a GDTL formula. Let $\Tilde{\Pi}$ be an additional finite set of atomic propositions such that $|\mathcal{L}_\varphi|=|\Tilde{\Pi}|$. Let $e:\mathcal{L}_\varphi\rightarrow\Tilde{\Pi}$ be a bijective map. Then GDTL formula $\varphi$ is equivalent to an LTL formula $\Tilde{\varphi}$ defined over atomic proposition set $\Pi\cup\Tilde{\Pi}$, with each predicate $l$ in $\mathcal{L}_\varphi$ being replaced by $e(l)\in\Tilde{\Pi}$.
\end{proposition}

Leveraging Proposition \ref{def:LTL equiv}, we can represent a GDTL formula using an equivalent DRA defined as follows:
\begin{definition}[Deterministic Rabin Automaton (DRA) \cite{baier2008principles}]\label{def:DRA}
A Deterministic Rabin Automaton (DRA) is a tuple $\mathcal{A}=(\mathcal{Q},\Sigma,\delta,q_0,F)$, where $\mathcal{Q}$ is a finite set of states, $\Sigma$ is the finite set of alphabet, $\delta:\mathcal{Q}\times\Sigma\rightarrow \mathcal{Q}$ is a finite set of transitions, $q_0\in \mathcal{Q}$ is the initial state, and $F=\{(B(1),C(1)),\ldots,(B(S),C(S))\}$ is a finite set of Rabin pairs such that
$B(s),C(s)\subseteq \mathcal{Q}$ for all $s=1,\ldots,S$ with $S$ being a positive integer.
\end{definition}

We consider that the system given by Eqn. \eqref{eq:state dynamics} and \eqref{eq:output} needs to satisfy a GDTL specification $\varphi$. We define a labeling function $L:\mathbb{B}\rightarrow 2^{\Tilde{\Pi}}$, where $\Tilde{\Pi}$ is the set of atomic propositions generated by Proposition \ref{def:LTL equiv}. For any $\pi\in\Tilde{\Pi}$, we define $\llbracket\pi\rrbracket$ as $\llbracket\pi\rrbracket=\{b:\pi\in L(b)\}$.
Consider $Z\in 2^{\Tilde{\Pi}}$. We define $\llbracket Z\rrbracket$ as
\begin{equation}\label{eq:subset atomic}
    \llbracket Z\rrbracket=\begin{cases}
    \mathbb{B}\setminus\cup_{\pi\in \Tilde{\Pi}}\llbracket\pi\rrbracket&\mbox{ if }Z=\emptyset\\
    \cap_{\pi\in Z}\llbracket\pi\rrbracket\setminus\cup_{\pi\in\Pi\setminus Z}\llbracket\pi\rrbracket&\mbox{ otherwise}
    \end{cases}
\end{equation}

We let $\mathbf{b}$ be a trajectory of belief state of infinite length. We now define the trace of $\mathbf{b}$ to establish the connection between the satisfaction of $\varphi$ to trajectory $\mathbf{b}$. 
\begin{definition}[Trace of Trajectory \cite{wongpiromsarn2015automata}]\label{def:trace}
An infinite sequence $Trace(\mathbf{b})=Z_0,Z_1,\ldots,Z_i,\ldots$, where $Z_i\in 2^{\Tilde{\Pi}}$ for all $i=0,1,\ldots$ is a trace of a trajectory $\mathbf{b}$ if there exists a sequence $t_0,t_1,\ldots,t_i,\ldots$ of time instants such that
\begin{enumerate}
    \item $t_0=0$
    \item $t_i\rightarrow\infty$ as $i\rightarrow\infty$
    \item $t_i<t_{i+1}$
    \item $\hat{x}_{t_i}\in \llbracket Z_i\rrbracket$
    \item if $Z_i\neq Z_{i+1}$, then there exists some $t_i'\in[t_i,t_{i+1}]$ such that $\hat{x}_t\in\llbracket Z_i\rrbracket$ for all $t\in(t_i,t_i')$, $\hat{x}_t\in\llbracket Z_{i+1}\rrbracket$ for all $t\in(t_i',t_{i+1})$, and either $\hat{x}_{t_i'}\in\llbracket Z_i\rrbracket$ or $\hat{x}_{t_i'}\in\llbracket Z_{i+1}\rrbracket$.
\end{enumerate}
\end{definition}
Using the trace of system trajectory, we have that GDTL specification $\varphi$ is satisfied if $Trace(\mathbf{b})\models \varphi$. The problem investigated in this paper is then formulated as follows:
\begin{problem}\label{prob:formulation}
Given a parameter $\epsilon\in(0,1)$, compute a control policy $\mu:\{y_{t'}:t'\in[0,t)\}\rightarrow\mathbb{R}^p$ for the system at each time $t$ mapping from the sequence of outputs $\{y_{t'}:t'\in[0,t)\}$ to a control input $u_t$, so that the probability of satisfying $\varphi$, denoted as $\mathbb{P}(\varphi)$, satisfies $\mathbb{P}(\varphi)\geq 1-\epsilon$ for any fault $r\in\mathcal{F}$ when fault $r$ occurs.
\end{problem}

\section{Solution Approach}\label{sec:sol}

In this section, we present the proposed solution approach. We first derive a fault-tolerant stochastic finite time convergence CBF. We then decompose the GDTL formula into a sequence of sub-formulae using the DRA representing $\varphi$. We show that the satisfaction probability of GDTL formula $\varphi$ can be expressed using that of each sub-formula. We then synthesize a control policy that enables the system to satisfy each sub-formula with certain probability guarantee using fault-tolerant stochastic finite time convergence CBF and fault-tolerant stochastic zeroing CBF. We finally derive the satisfaction probability guarantee for each sub-formula.

\subsection{Fault-Tolerant Stochastic Finite Time Convergence CBF}\label{sec:FT-SFCBF}

This subsection derives fault-tolerant stochastic finite time convergence CBF. We first define stochastic finite time convergence CBF under complete information in the absence of an adversary. Then we extend it to the incomplete information setting. We finally propose the fault-tolerant stochastic finite time convergence CBF to address the presence of the adversary.

\begin{definition}[Stochastic Finite Time Convergence CBF under Complete Information]\label{def:FCBF 1}
Consider the dynamical system given in Eqn. \eqref{eq:state dynamics}
whose solution is denoted as $x(t;x_0)$. Consider a set $\mathcal{C}=\{x:h(x)\geq 0\}$ where $h:\mathbb{R}^n\rightarrow \mathbb{R}$ is a twice continuously differentiable function. We say $h$ is a stochastic finite time convergence CBF if the following conditions hold:
\begin{enumerate}
    \item There exist class $\mathcal{K}_\infty$ functions $\alpha_1$ and $\alpha_2$ such that
    \begin{equation}
        \alpha_1(h(x))\leq h(x)\leq \alpha_2(h(x))
    \end{equation}
    for all $x$.
    \item There exist control input $u$ such that 
    \begin{multline}\label{eq:CLF}
        \frac{\partial h}{\partial x}(x)(f(x)+g(x)u)+\frac{1}{2}tr\left(\sigma^\top\frac{\partial^2 h(x)}{\partial x^2}\sigma\right)\\
        + sgn(h(x)) |h(x)|\geq 0.
    \end{multline}
\end{enumerate}
\end{definition}

Stochastic finite time convergence CBF provides the \emph{finite time stability in probability} guarantee as given by Definition \ref{def:finite time stability}. We formally state this result as follows:
\begin{lemma}\label{lemma:SFCBF complete info}
Let $\mathcal{C}=\{x:h(x)\geq 0\}$. Suppose the system admits a stochastic finite time convergence CBF $h(x)$. Let $T=\inf\{t:x(t;x_0)\in\mathcal{C}\}$. Then $\mathbb{P}(T<\infty)=1$, i.e., the system reaches $\mathcal{C}$ within finite time almost surely. Moreover, the system remains in $\mathcal{C}$ for all $t'\geq T$ almost surely.
\end{lemma}
\begin{proof}


We define a function $V(h(x))$ as $V(h(x))=h(x)^2$. 
Using the definition of $V$, we have that
\begin{align*}
    &\frac{\partial V}{\partial x}(x)(f(x)+g(x)u)+\frac{1}{2}tr\left(\sigma^\top\frac{\partial^2 V(x)}{\partial x^2}\sigma\right)\\
    =&2h(x)\left[\frac{\partial h}{\partial x}(x)(f(x)+g(x)u)+\frac{1}{2}tr\left(\sigma^\top\frac{\partial^2 h(x)}{\partial x^2}\sigma\right)\right]
\end{align*}
If Eqn. \eqref{eq:CLF} holds and $x\notin \mathcal{C}$, we have that
\begin{multline*}
    \frac{\partial V}{\partial x}(x)(f(x)+g(x)u)+\frac{1}{2}tr\left(\sigma^\top\frac{\partial^2 V(x)}{\partial x^2}\sigma\right) \\
    \leq 2h(x) |V(x)|^{\frac{1}{2}}\leq 0.
\end{multline*}
Note that $h(x)<0$ when $x\notin \mathcal{C}$. By Theorem \ref{thm:SCLF}, we have that $\mathbb{P}(T<\infty)=1$, where $T=\inf\{t:h(x(t;x_0))=0\}$ is the stochastic settling time. Hence the system reaches $\mathcal{C}$ almost surely when $x_0\notin \mathcal{C}$. When $x_0\in\mathcal{C}$, the lemma holds since $T=0$.

When Eqn. \eqref{eq:CLF} holds, we have that $x_t\in\mathcal{C}$ for $t\geq T$ and
\begin{equation*}
    \frac{\partial h}{\partial x}(x)(f(x)+g(x)u)+\frac{1}{2}tr\left(\sigma^\top\frac{\partial^2 h(x)}{\partial x^2}\sigma\right)\geq -h(x).
\end{equation*}
By \cite[Thm. 3]{clark2021control}, we have that $\mathbb{P}(x_t\in\mathcal{C},~\forall t\geq T)=1$.
\end{proof}

In the following, we extend Definition \ref{def:FCBF 1} and Lemma \ref{lemma:SFCBF complete info} to the incomplete information setting where the system state is estimated via an EKF. Define 
\begin{equation*}
    \Bar{h}_{\varepsilon} = \sup\{h(x):\|x-x'\|_2\leq \varepsilon \text{ for some }x'\in h^{-1}(\{0\})\}
\end{equation*}
as the supremum of $h(x)$ for all $x$ belonging to the $\varepsilon$-neighborhood of the boundary of $\mathcal{C}$. We then present the following preliminary result \cite{clark2021control}.
\begin{lemma}[\cite{clark2021control}]\label{lemma:incomplete info safety}
If $\|x_t-\hat{x}_t\|_2\leq \varepsilon$ for all $t$ and $h(\hat{x}_t)>\Bar{h}_\varepsilon$ for all $t$, then $x_t\in\mathcal{C}$ for all $t$.
\end{lemma}

We define $\hat{h}(x)=h(x)-\Bar{h}_\varepsilon$. Provided Lemma \ref{lemma:incomplete info safety}, we are now ready to develop stochastic finite time convergence CBF under incomplete information setting for the system described by \eqref{eq:state dynamics} and \eqref{eq:output} employing EKF for state estimation, assuming that $a_t=\mathbf{0}$ for all $t\geq 0$.
\begin{theorem}\label{thm:SFCBF}
Let $T=\inf\{t:x(t;x_0)\in\mathcal{C}\}$. Suppose there exists a twice continuously differentiable function $h:\mathbb{R}^n\rightarrow\mathbb{R}$, control input $u$, and class $\mathcal{K}$ functions $\alpha_1$, $\alpha_2$ such that
\begin{subequations}
\begin{align}
    &\alpha_1(\hat{h}(x))\leq h(x)\leq \alpha_2(\hat{h}(x))\\
    &\frac{\partial h}{\partial x}(\hat{x})(f(\hat{x})+g(\hat{x})u)-\varepsilon \|\frac{\partial h}{\partial x}Kc\|_2\nonumber\\
    &+\frac{1}{2}tr\left(\sigma_t^\top K^\top\frac{\partial^2 h}{\partial x^2}K\sigma\right)+  sgn(\hat{h}(\hat{x}))|\hat{h}(\hat{x})|\geq 0\label{eq:SFCBF incomplete info}
\end{align}
\end{subequations}
then $\mathbb{P}(T<\infty|\|x_t-\hat{x}_t\|_2\leq \varepsilon,~\forall t)=1$, i.e., the system reaches $\mathcal{C}$ within finite time almost surely.
\end{theorem}
\begin{proof}
Since we use EKF for state estimation, we have that the state estimate and covariance follow Eqn. \eqref{eq:EKF estimate} and \eqref{eq:EKF Ricatti}. Given the output dynamics of $y_t$, we have that
\begin{align*}
    d\hat{x}_t&=(f(\hat{x}_t) + g(\hat{x}_t)u_t)d_t + K_t(cx_tdt+\nu_tdV_t-c\hat{x}_tdt)\\
    &=(f(\hat{x}_t) + g(\hat{x}_t)u_t + K_tc(x_t-\hat{x}_t))dt + K_t\nu_tdV_t.
\end{align*}
Let $\hat{h}(x)=h(x)-\Bar{h}_\varepsilon$. We have that
\begin{multline}
    d\hat{h}_t = \bigg(\frac{\partial h}{\partial x}(\hat{x}_t)\left(f(\hat{x}_t) + g(\hat{x}_t)u_t + K_tc(x_t-\hat{x}_t)\right)\\
    +\frac{1}{2}tr\left(\nu_t^\top K_t^\top \frac{\partial^2 h}{\partial x}(\hat{x}_t)K_t\nu_t\right)\bigg)dt - \frac{\partial h}{\partial x}K_t\nu_tdV_t.
\end{multline}
When $\|x_t-\hat{x}_t\|_2\leq \varepsilon$ holds, then 
\begin{equation*}
    \frac{\partial h}{\partial x}K_tc(x_t-\hat{x}_t)\geq  -\left\|\frac{\partial h}{\partial x}K_tc\right\|_2\|x_t-\hat{x}_t\|_2\geq -\varepsilon\left\|\frac{\partial h}{\partial x}K_tc\right\|_2.
\end{equation*}
When Eqn. \eqref{eq:SFCBF incomplete info} holds, we thus have that
\begin{align*}
    &\frac{\partial h}{\partial x}(\hat{x}_t)\left(f(\hat{x}_t) + g(\hat{x}_t)u_t + K_tc(x_t-\hat{x}_t)\right)\\
    &\quad\quad\quad+\frac{1}{2}tr\left(\nu_t^\top K_t^\top \frac{\partial^2 h}{\partial x^2}(\hat{x}_t)K_t\nu_t\right)\\
    \geq&\frac{\partial h}{\partial x}(\hat{x}_t)\left(f(\hat{x}_t) + g(\hat{x}_t)u_t\right) - \varepsilon\left\|\frac{\partial h}{\partial x}K_tc\right\|_2\\
    &\quad\quad\quad+\frac{1}{2}tr\left(\nu_t^\top K_t^\top \frac{\partial^2 h}{\partial x^2}(\hat{x}_t)K_t\nu_t\right)\\
    \geq & - sgn(\hat{h}(\hat{x}_t))|\hat{h}(\hat{x}_t)|.
\end{align*}
Finally, according to Lemma \ref{lemma:SFCBF complete info}, we have that time $T=\inf\{t:x(t;x_0)\in\mathcal{C}\}$ is finite with probability $\mathbb{P}(T<\infty|\|x_t-\hat{x}_t\|_2\leq\varepsilon,~\forall t)=1$.
\end{proof}

We are now ready to present fault-tolerant stochastic finite time convergence CBF for the system given by \eqref{eq:state dynamics}-\eqref{eq:output} as follows. The system utilizes $m$ EKFs, with each associated with one attack pattern $r_i\in\mathcal{F}$, where $i=1,\ldots,m$.
\begin{lemma}\label{thm:FT-SFCBF}
Let $\Bar{h}_{\varepsilon_i} = \sup\{h(x):\|x-x'\|_2\leq \varepsilon_i \text{ for some }x'\in h^{-1}(\{0\})\}$ and $\hat{h}_i(x)=h(x)-\Bar{h}_{\varepsilon_i}$. Let $T=\inf\{t:x(t;x_0)\in\mathcal{C}\}$. Suppose $\varepsilon_1,\ldots,\varepsilon_m$ and $\theta_{ij}$ are chosen such that there exists $\varrho>0$, making any $X_t'\subseteq X_t(\varrho)$ satisfying $\|\hat{x}_{t,i}-\hat{x}_{t,j}\|\leq \theta_{ij}$ for all $i,j\in X_t'$ there exists $u$ satisfying
    \begin{equation}
        \Lambda_i(\hat{x}_{t,i})u>0,~\forall i\in X_t'
    \end{equation}
    where $\Lambda_i(\hat{x}_{t,i})=\frac{\partial h_i}{\partial x}(\hat{x}_{t,i})g(\hat{x}_{t,i})$, $X_t(\varrho)=\{i:\hat{h}_i(\hat{x}_{t,i})<\varrho\}$,
then 
\begin{multline*}
    \mathbb{P}(T<\infty|\|\hat{x}_{t,i}-\hat{x}_{t,i,j}\|_2\leq\theta_{ij}/2,~\forall j,\\
    \|\hat{x}_{t,i}-x_t\|_2\leq \varepsilon_i,~\forall t,\forall r_i\in\mathcal{F})=1.
\end{multline*}
\end{lemma}
The proof of Lemma \ref{thm:FT-SFCBF} is omitted due to space constraint. Functions $h_1,\ldots,h_m$ satisfying conditions 1) and 2) in Lemma \ref{thm:FT-SFCBF} are fault tolerant stochastic finite time convergence CBFs.

\subsection{Decomposition of GDTL Formula $\varphi$ and Satisfaction of the Decomposed Sub-formulae}

Given the GDTL formula, we construct the DRA representing it, as defined in Definition \ref{def:DRA}. We then pick an accepting run of the automaton $\mathcal{A}$, denoted as $\eta$. We rewrite $\eta$ into the prefix suffix form $\eta=q_1,\ldots,q_n,(\eta_{suff})^\omega$. We next decompose the accepting run $\eta$ into a sequence of sub-formulae, denoted as $\psi_1,\ldots,\psi_N$. Here 
\begin{equation}\label{eq:sub-formula}
    \psi_j = \Phi_j\mathcal{U}\Box (\phi_j\land \Phi_{j+1}),
\end{equation}
where $\Phi_j$ is the input word for self-loop transition $\delta(q_j,\Phi_j)=q_j$, and $\phi_j$ is the input word corresponding to the transition from state $q_j$ to $q_{j+1}$. That is, sub-formula $\psi_j$ models a one-step transition from state $q_j$ to $q_{j+1}$ along accepting run $\eta$.

We next synthesize the control policy to satisfy each sub-formula leveraging fault-tolerant finite time convergence control barrier functions in Section \ref{sec:FT-SFCBF} and fault-tolerant zeroing control barrier functions proposed in \cite{clark2020control}.
We then derive the satisfaction probability of each sub-formula, which is used to bound the probability of satisfying $\varphi$.

We let functions $h^j$ and $d^j$ be given as
\begin{align*}
    &\{b:h^j(b)\geq 0\}=\{b:(L(b)\models \phi_j\land\Phi_{j+1})\lor (L(b)\models \Phi_j )\}\\
    &\{b:d^j(b)\geq 0\}=\{b:L(b)\models \Phi_{j+1}\}.
\end{align*}
Then sub-formula $\psi_j$ indicates that the system trajectory needs to guarantee $h^j(b)\geq 0$ for all $[t_j,t_{j+1}]$ and $d^j(b)\geq 0$ for some $t'\in[t_j,t_{j+1}]$, where $t_j$ and $t_{j+1}$ are some time instants satisfying $t_{j+1}\geq t_j\geq 0$. We now define 
\begin{align}
    &\Bar{h}^j_{\varepsilon} = \sup\{h^j(b):\|b-b'\|_2\leq \varepsilon\nonumber\\
    &\quad\quad\quad\text{ for some }b'\in (h^j)^{-1}(\{0\})\}\label{eq:h bar}\\
    &\Bar{d}^j_{\varepsilon} = \sup\{d^j(b):\|b-b'\|_2\leq \varepsilon \nonumber\\
    &\quad\quad\quad\text{ for some }b'\in (d^j)^{-1}(\{0\})\}\label{eq:d bar}\\
    &\hat{h}^j(b)=h^j(b)-\Bar{h}^j_{\varepsilon},\quad\quad\hat{d}^j(b)=d^j(b)-\Bar{d}^j_{\varepsilon}.\label{eq:h d hat}
\end{align}

We next present an algorithm for synthesizing the control policy so that the probability of satisfying sub-formula $\psi_j$ is lower bounded in Algorithm \ref{algo:synthesis}. In line 3, we compute $X_t(\varrho)=\{i:\hat{h}_i^j(b_{t,i})<\varrho_1,\hat{d}_i^j(b_{t,i})<\varrho_2\}$, where $\varrho=[\varrho_1,\varrho_2]^\top$ and $\varrho_1,\varrho_2>0$ are constants. In line 4, we first compute the set of control inputs $\Omega_i$ and $\Gamma_i$, where $\Omega_i$ ensures that $h^j(b_{t,i})\geq 0$ for all $t\in[t_j,t_{j+1}]$, and $\Gamma_i$ ensures that there exists $t'\in [t_j,t_{j+1}]$ such that $d^j(x)\geq 0$, provided that the fault pattern is $r_i$. If $\cap_{i\in X_t(\varrho)}(\Omega_i\cap \Gamma_i)\neq \emptyset$, then any control input $u\in \cap_{i\in X_t(\varrho)}(\Omega_i\cap \Gamma_i)$ guarantees the satisfaction of sub-formula $\psi_j$ regardless of the fault pattern. If such $u$ does not exist, then line 5-11 aims to identify the fault pattern by comparing the state estimate when fault pattern $r_i$ is removed.
Line 17 computes the control input by solving a quadratic program. If no feasible control input can be found, we will repeat the procedure for another accepting run $\eta'\neq \eta$.

\begin{theorem}\label{thm:satisfaction of sub-formula}
Let $b_{t,i}= (\hat{x}_{t,i},P_{t,i})$. We define 
\begin{align*}
    &\Lambda_i^j(b_{t,i},u)=\frac{\partial h_i^j}{\partial x}(b_{t,i})g(\hat{x}_{t,i})u\nonumber\\
    &~+ tr\left[\frac{\partial h_i^j}{\partial P}(b_{t,i})\left(\frac{\partial g}{\partial x}(\hat{x}_{t,i})uP_{t,i}+P_{t,i}(\frac{\partial g}{\partial x}(\hat{x}_{t,i})u)^\top\right)\right],\\ &\Xi_i^j(b_{t,i},u)=\frac{\partial d_i^j}{\partial x}(b_{t,i})g(\hat{x}_{t,i})u\nonumber\\
    &~+ tr\left[\frac{\partial d_i^j}{\partial P}(b_{t,i})\left(\frac{\partial g}{\partial x}(\hat{x}_{t,i})uP_{t,i}+P_{t,i}(\frac{\partial g}{\partial x}(\hat{x}_{t,i})u)^\top\right)\right].
\end{align*}
Suppose $\varepsilon_1,\ldots,\varepsilon_m$ and $\theta_{ij}$ are chosen such that there exists $\varrho$ making $\|b_{t,i}-b_{t,k}\|\leq \theta_{ik}$ holds for all $i,k\in X_t'$ and for any $X_t'\subseteq X_t(\varrho)$. If there exists some $u$ satisfying
    \begin{equation}\label{eq:control set}
        \Lambda_i(b_{t,i},u)> 0,\quad \Xi_i(b_{t,i},u)> 0,~\forall i\in X_t'.
    \end{equation}
then 
\begin{multline*}
    \mathbb{P}(\psi_j|
    \|b_{t,i}-b_{t,i,k}\|_2\leq\theta_{ik}/2,~\forall j,\\
    \|b_{t,i}-b_t\|_2\leq \varepsilon_i,~\forall t,\forall r_i\in\mathcal{F})=1.
\end{multline*}
\end{theorem}
\begin{algorithm}[!hbp]
	\caption{Control synthesis for sub-formula $\psi_j$.}
	\label{algo:synthesis}
	\begin{algorithmic}[1]
		\State \textbf{Input}: Sub-formula $\psi_j$, system dynamics \eqref{eq:state dynamics}-\eqref{eq:output}, parameters $\varepsilon_1,\ldots,\varepsilon_m$, $\theta_{ij}$ where $i<j$, and $\varrho_1,\varrho_2>0$.
		\State \textbf{Output:} control input $u_t$
		\State Compute $X_t(\varrho)=\{i:\hat{h}_i^j(b_{t,i})<\varrho_1,\hat{d}_i^j(b_{t,i})<\varrho_2\}$, where $\varrho=[\varrho_1,\varrho_2]^\top$
		\State For each $i\in X_t(\varrho)$, compute
		\begin{align}
		    &\Omega_i=\Big\{u:\frac{\partial h_i^j}{\partial x}(f(\hat{x}_{t,i})+g(\hat{x}_{t,i})u) - \varepsilon\|\frac{\partial h_i^j}{\partial x}K_{t,i}c\|_2\nonumber
            \\
            &+ \frac{1}{2}tr\left(\nu_{t,i}^\top K_{t,i}^\top\frac{\partial^2 h_i^j}{\partial x^2}K_{t,i}\nu_{t,i}\right)\geq -\hat{h}(b_{t,i})\Big\}\\
            &\Gamma_i=\Big\{u:\frac{\partial d_i^j}{\partial x}(f(\hat{x}_{t,i})+g(\hat{x}_{t,i})u) - \varepsilon\|\frac{\partial d_i^j}{\partial x}K_{t,i}c\|_2\nonumber
            \\
            &+ \frac{1}{2}tr\left(\nu_{t,i}^\top K_{t,i}^\top\frac{\partial^2 d_i^j}{\partial x^2}K_{t,i}\nu_{t,i}\right)\nonumber\\
            &\quad\quad\geq - sgn(\hat{d}_i^j(b_{t,i}))|\hat{d}(b_{t,i})|\Big\}
		\end{align}
		\While{$\cap_{i\in X_t(\varrho)}(\Omega_i\cap \Gamma_i)= \emptyset$ and $X_t(\varrho)\neq \emptyset$}
		\For {$i,k\in \{1,\ldots,m\}$}
        \If{$\|\hat{x}_{t,i}-\hat{x}_{t,k}\|_2>\theta_{ik}$ and $\|\hat{x}_{t,i}-\hat{x}_{t,i,k}\|_2>\theta_{ik}/2$}
        \State Update $X_t(\varrho)\leftarrow X_t(\varrho)\setminus\{i\}$
        \EndIf
        \EndFor
		\EndWhile
		\While{$\cap_{i\in X_t(\varrho)}(\Omega_i\cap \Gamma_i)= \emptyset$ and $X_t(\varrho)\neq \emptyset$}
		\State $k^*\leftarrow\argmax\{y_{t,k}-c_k\hat{x}_{t,k}:k\in X_t(\varrho)\}$
		\State Update $X_t(\varrho)\leftarrow X_t(\varrho)\setminus\{k^*\}$
		\EndWhile
		\If{$\cap_{i\in X_t(\varrho)}(\Omega_i\cap \Gamma_i)\neq \emptyset$}
		\State Solve $u_t$ as $u_t=\underset{u_t\in\cap_{i\in X_t(\varrho)}(\Omega_i\cap \Gamma_i)}{\arg\min} {u_t}^\top Mu_t$, where $M$ is positive definite
		\EndIf
		\State \textbf{Return $u_t$}
	\end{algorithmic}
\end{algorithm}

\begin{proof}
We show that if $\|b_{t,i}-b_t\|_2\leq \varepsilon_i$ and $\|b_{t,i}-b_{t,i,k}\|_2\leq\theta_{ik}/2$ for all $t$, then $u_t\in\Omega_i\cap\Gamma_i$ holds when $\hat{h}_i^j(b_{t,i})<\varrho_1$ and $\hat{d}_i^j(b_{t,i})<\varrho_2$.

Suppose that the fault pattern $r=r_i$. At time $t$, suppose that $\hat{h}_i^j(b_{t,i})<\varrho_1$ and $\hat{d}_i^j(b_{t,i})<\varrho_2$ hold. In addition, suppose that $\|b_{t,i}-b_{t,i,k}\|_2\leq\theta_{ik}/2$. We divide our discussion into three cases.

\underline{\emph{Case I.}} We first consider $\|b_{t,i}-b_{t,k}\|_2\leq\theta_{ik}$ for all $i,k\in X_t(\varrho)$. Note that $\Omega_i$ and $\Gamma_i$ can be represented as
\begin{equation*}
    \Omega_i=\{u:\Lambda_i(b_{t,i},u)\geq \omega_i\},\quad \Gamma_i=\{u:\Xi_i(b_{t,i},u)\geq \iota_i\},
\end{equation*}
where $\omega_i,\iota_i\in\mathbb{R}$ are some constants. When $\|b_{t,i}-b_{t,k}\|_2\leq\theta_{ik}$ for all $i,k\in X_t(\varrho)$ holds, by condition 1), we then have that there exists $u_t$ satisfying
\begin{equation*}
    \Lambda_i(b_{t,i},u_t)> 0,\quad \Xi_i(b_{t,i},u_t)> 0,~\forall i\in X_t'.
\end{equation*}
indicating that $u_t\in \cap_{i\in X_t(\varrho)}(\Omega_i\cap \Gamma_i)$.

\underline{\emph{Case II.}} We next consider that $\|b_{t,i}-b_{t,k}\|_2\leq\theta_{ik}$ for all $k\in X_t(\varrho)$, but there exists some $z,k\in X_t(\varrho)\setminus\{i\}$ such that $\|b_{t,k}-b_{t,z}\|>\theta_{zk}$. In this case, by the iteration between line 5-11 in Algorithm \ref{algo:synthesis}, $\Omega_k\cap\Gamma_k$ is removed. By the hypothesis that $\|b_{t,i}-b_{t,i,k}\|_2\leq\theta_{ik}$ for all $k\in X_t(\varrho)$, we have that the updated $X_t(\varrho)$ now reduces to Case I, and thus our previous analysis can be applied.

We finally suppose that $\|b_{t,i}-b_{t,k}\|_2>\theta_{ik}$ for some $k\in X_t(\varrho)$. We then have that
\begin{subequations}
\begin{align*}
    \theta_{ik}&\leq \|b_{t,i}-b_{t,i,k}+b_{t,i,k}-b_{t,k}\|_2\\
    &\leq \|b_{t,i}-b_{t,i,k}\|_2+\|b_{t,i,k}-b_{t,k}\|_2\\
    &\leq \frac{\theta_{ik}}{2}+\|b_{t,i,k}-b_{t,k}\|_2
\end{align*}
\end{subequations}
where the second inequality holds by triangle inequality, and the last inequality holds by the hypothesis that $\|b_{t,i}-b_{t,i,k}\|_2\leq\theta_{ik}/2$. In this case, $\|b_{t,i,k}-b_{t,k}\|_2\geq \theta_{ik}/2$, indicating that $\{k\}$ is removed by line 8 of Algorithm \ref{algo:synthesis}. Then our previous analysis becomes applicable once all such $k$ are removed. Lastly, using Theorem \ref{thm:SFCBF} and \cite{clark2020control}, we have that the control input $u$ satisfies sub-formula $\psi_j$ almost surely.
\end{proof}

We conclude this section by quantifying the probability of satisfying GDTL formula $\varphi$ as follows.
\begin{theorem}\label{thm:uniform bound}
If the conditions in Theorem \ref{thm:satisfaction of sub-formula} are satisfied,
then the probability of satisfying GDTL specification $\varphi$ satisfies
\begin{multline}\label{eq:uniform bound}
    \mathbb{P}( \varphi|\|b_{t,i}-b_{t,i,k}\|_2\leq\theta_{ik}/2,~\forall j,\\
    \|b_{t,i}-b_t\|_2\leq \varepsilon_i,~\forall t,\forall r_i\in\mathcal{F})=1.
\end{multline}
\end{theorem}
\begin{proof}
We denote the probability that $\psi_j,\ldots,\psi_{j'}$ occur in order as $\mathbb{P}(\psi_j,\ldots,\psi_{j'})$, where $j\leq j'$. We also denote the event that $\|b_{t,i}-b_{t,i,k}\|_2\leq\theta_{ik}/2$ holds for all $j$ and $ \|b_{t,i}-b_t\|_2\leq \varepsilon_i$ holds for all $t$ under fault $r_i$ as $\Upsilon_i$.
By Bayes' theorem, we have that
\begin{align*}
    &\mathbb{P}( \varphi|\Upsilon_i,\forall r_i\in\mathcal{F})\geq \mathbb{P}( \psi_1\ldots\psi_N|\Upsilon_i,\forall r_i\in\mathcal{F})\\
    =&\mathbb{P}( \psi_N|\psi_1\ldots\psi_{N-1},\Upsilon_i,~\forall r_i\in\mathcal{F})\nonumber\\
    &\quad\cdot\mathbb{P}( \psi_{N-1}|\psi_1\ldots\psi_{N-2},\Upsilon_i,~\forall r_i\in\mathcal{F})\ldots\nonumber\\
    &\quad\cdot\mathbb{P}( \psi_1|\Upsilon_i,~\forall r_i\in\mathcal{F})\mathbb{P}(\Upsilon_i,~\forall r_i\in\mathcal{F}).
\end{align*}
By our hypothesis, we have that $\mathbb{P}(\Upsilon_i,~\forall r_i\in\mathcal{F})=1$. Using Theorem \ref{thm:satisfaction of sub-formula}, we have that $\mathbb{P}( \psi_j|\psi_1\ldots\psi_{j-1},\Upsilon_i,~\forall r_i\in\mathcal{F})=1$
holds for all $j=1,\ldots,N$. Therefore, we have that Eqn. \eqref{eq:uniform bound} holds, completing the proof.
\end{proof}

Based on Theorem \ref{thm:uniform bound}, we have that if the EKFs guarantee that $\mathbb{P}(\Upsilon_i)\geq 1-\epsilon$ for all $r_i\in\mathcal{F}$, then the probability of satisfying GDTL specification $\varphi$ can be bounded as $\mathbb{P}(\varphi)\geq 1-\epsilon$, as desired by Problem \ref{prob:formulation}. The accuracy of EKFs can be obtained via Lemma \ref{lemma:estimation error}.

\section{Case Study}\label{sec:simulation}

In this section, we present a case study on the coordination of two wheeled mobile robots (WMRs) in a 2-D domain. Each WMR follows dynamics:
\begin{equation}
\label{original_model}
    \begin{bmatrix}
    [\dot{x}_{t}]_1\\
    [\dot{x}_t]_2\\
    \dot{z}_t
    \end{bmatrix}=\begin{bmatrix}
    \cos z_t &0\\
    \sin z_t &0\\
    0&1
    \end{bmatrix}\begin{bmatrix}
    v_t\\\omega_t
    \end{bmatrix}+w_t
\end{equation}
where $[x_t]_1$ and $[x_t]_2$ model the position of the robot at time $t$, $z_t$ is the orientation of the robot at time $t$, $u_t=[v_t,\omega_t]^\top$ is the control input consisting of the linear and angular velocities, and $w$ is a zero-mean Gaussian process noise. Following \cite{chen2018building}, we apply feedback linearization to Eqn. \eqref{original_model} and use $7$ sensors $y_t\in\mathbb{R}^7$ to measure the location and velocities for each WMR. In particular, we let $[y_t]_1$ and $[y_t]_2$ measure horizontal position $[x_t]_1$, and use $[y_t]_3$ and $[y_t]_4$ to measure vertical position $[x_t]_2$. Outputs $[y_t]_5$ and $[y_t]_6$ measure the velocities of the WMR. The orientation of the WMR is measured by $[y_t]_7$.

The WMRs need to coordinate to satisfy a GDTL specification $\varphi$ given as $\varphi= \land_{i=1}^3 \tilde{\varphi}_i$, where $\tilde{\varphi}_1 = \Diamond(dest1a \land \Diamond dest1b)\land\Diamond dest2$, $ \tilde{\varphi}_2 = \Box(\neg Obs)$, and $\tilde{\varphi}_3=\Box(tr(P)\leq 0.9)$.
Specification $\varphi$ requires the WMRs to reach a set of destinations, denoted as $dest1a, dest1b$, and  $dest2$, in order, and to avoid the unsafe region $obs$. In addition, the WMRs must maintain an uncertainty that is less than $0.9$. There exists a malicious adversary that manipulates the outputs of $[y_t]_2$ and $[y_t]_4$ to bias the WMRs' location estimates.

\begin{figure}
    \centering
    \includegraphics[scale = 0.42]{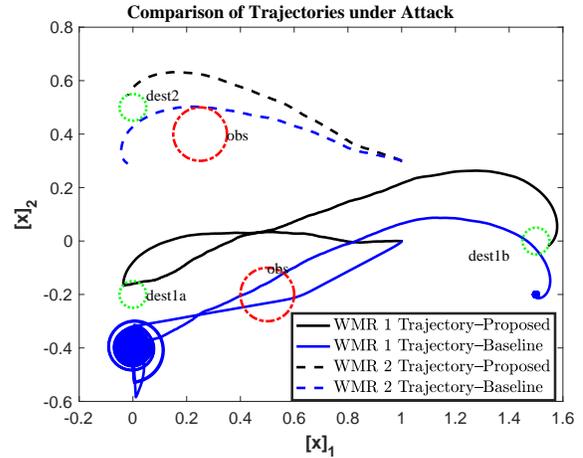}
    \caption{The trajectories of two WMRs generated using our proposed approach and a baseline are plotted in black and blue colors, respectively. The baseline applies a CBF-based control policy with state estimate being calculated using all sensors. Our proposed approach satisfies the given GDTL specification $\varphi$, whereas the baseline violates it.}
    \vspace*{-2mm}
    \label{fig:traj}
\end{figure}

We compare our proposed approach with a baseline approach \cite{niu2020control}. The baseline approach adopts the state-of-the-art abstraction-free CBF-based control synthesis using the state estimate obtained by all sensor measurements.

In the following, we demonstrate our proposed approach. We generate the DRA representing GDTL specification $\varphi$. We then pick an accepting run $\eta$ and decompose $\varphi$ into a sequence of sub-formulae using $\eta$. In this case study, accepting run $\eta$ requires the WMRs to reach destinations $dest1a$, $dest1b$, and $dest2$ in this order, while guaranteeing $\tilde{\varphi}_2$ and $\tilde{\varphi}_3$ to be always satisfied.

We present the trajectories for both robots using our proposed approach and the baseline in Fig. \ref{fig:traj} using lines in black and blue colors, respectively. The trajectories of WMR $1$ and $2$ are presented using solid and dashed lines, respectively. We observe that the trajectories of the robots using our proposed approach never enter the unsafe region $obs$ and reach their assigned destinations. In addition, our proposed approach yields $tr(P)=0.03\leq 0.9$ and thus satisfies specification $\varphi_3$. However, the trajectories generated by the baseline approach enter the the unsafe region $obs$ and thus violate specification $\varphi_2$. Moreover, WMR $1$ is biased by the sensor attack and fails to reach its destination $dest1a$ using the baseline approach. To summarize, our proposed approach guarantees that WMR $1$ and $2$ can coordinate to satisfy GDTL formula $\varphi$ in the presence of sensor attacks.
\section{Conclusion}\label{sec:conclusion}

In this paper, we studied the problem of abstraction-free synthesis for a control-affine nonlinear system to satisfy a Gaussian distribution temporal logic specification in the presence of process and observation noises as well as sensor faults and attacks. We proposed a solution approach that decomposes the given specification into a sequence of reach-avoid tasks that needs to be satisfied by the system. We developed the sufficient conditions for the controller so that each decomposed sub-formula can be satisfied by the system in the presence of sensor faults and attacks using our proposed fault-tolerant finite time convergence control barrier functions, along with the fault-tolerant zeroing control barrier functions. We proved that our synthesized controller guarantees almost sure satisfaction for each decomposed sub-formula, and thus ensures that the system satisfies the Gaussian distribution temporal logic specification with probability one, provided that the extended Kalman filters used by the system ensures certain accuracy bounds. We presented a numerical case study on the coordination of two wheeled mobile robots to demonstrate our proposed approach.

\bibliographystyle{IEEEtran}
\bibliography{IEEEabrv,MyBib}

\begin{thebibliography}{10}
\providecommand{\url}[1]{#1}
\csname url@rmstyle\endcsname
\providecommand{\newblock}{\relax}
\providecommand{\bibinfo}[2]{#2}
\providecommand\BIBentrySTDinterwordspacing{\spaceskip=0pt\relax}
\providecommand\BIBentryALTinterwordstretchfactor{4}
\providecommand\BIBentryALTinterwordspacing{\spaceskip=\fontdimen2\font plus
\BIBentryALTinterwordstretchfactor\fontdimen3\font minus
  \fontdimen4\font\relax}
\providecommand\BIBforeignlanguage[2]{{%
\expandafter\ifx\csname l@#1\endcsname\relax
\typeout{** WARNING: IEEEtran.bst: No hyphenation pattern has been}%
\typeout{** loaded for the language `#1'. Using the pattern for}%
\typeout{** the default language instead.}%
\else
\language=\csname l@#1\endcsname
\fi
#2}}

\bibitem{baier2008principles}
C.~Baier, J.-P. Katoen, and K.~G. Larsen, \emph{{Principles of Model
  Checking}}.\hskip 1em plus 0.5em minus 0.4em\relax MIT Press, 2008.

\bibitem{vasile2016control}
C.-I. Vasile, K.~Leahy, E.~Cristofalo, A.~Jones, M.~Schwager, and C.~Belta,
  ``Control in belief space with temporal logic specifications,'' in \emph{55th
  IEEE Conference on Decision and Control (CDC)}.\hskip 1em plus 0.5em minus
  0.4em\relax IEEE, 2016, pp. 7419--7424.

\bibitem{leahy2019control}
K.~Leahy, E.~Cristofalo, C.-I. Vasile, A.~Jones, E.~Montijano, M.~Schwager, and
  C.~Belta, ``Control in belief space with temporal logic specifications using
  vision-based localization,'' \emph{The International Journal of Robotics
  Research}, vol.~38, no.~6, pp. 702--722, 2019.

\bibitem{kress2009temporal}
H.~Kress-Gazit, G.~E. Fainekos, and G.~J. Pappas, ``Temporal-logic-based
  reactive mission and motion planning,'' \emph{IEEE Transactions on Robotics},
  vol.~25, no.~6, pp. 1370--1381, 2009.

\bibitem{coogan2015traffic}
S.~Coogan, E.~A. Gol, M.~Arcak, and C.~Belta, ``Traffic network control from
  temporal logic specifications,'' \emph{IEEE Transactions on Control of
  Network Systems}, vol.~3, no.~2, pp. 162--172, 2015.

\bibitem{fainekos2009temporal}
G.~E. Fainekos, A.~Girard, H.~Kress-Gazit, and G.~J. Pappas, ``Temporal logic
  motion planning for dynamic robots,'' \emph{Automatica}, vol.~45, no.~2, pp.
  343--352, 2009.

\bibitem{ding2014optimal}
X.~Ding, S.~L. Smith, C.~Belta, and D.~Rus, ``Optimal control of {M}arkov
  decision processes with linear temporal logic constraints,'' \emph{IEEE
  Transactions on Automatic Control}, vol.~59, no.~5, pp. 1244--1257, 2014.

\bibitem{wongpiromsarn2009receding}
T.~Wongpiromsarn, U.~Topcu, and R.~M. Murray, ``Receding horizon temporal logic
  planning for dynamical systems,'' in \emph{48th IEEE Conference on Decision
  and Control (CDC)}.\hskip 1em plus 0.5em minus 0.4em\relax IEEE, 2009, pp.
  5997--6004.

\bibitem{wolff2012robust}
E.~M. Wolff, U.~Topcu, and R.~M. Murray, ``Robust control of uncertain {M}arkov
  decision processes with temporal logic specifications,'' in \emph{51st
  Conference on Decision and Control}.\hskip 1em plus 0.5em minus 0.4em\relax
  IEEE, 2012, pp. 3372--3379.

\bibitem{horowitz2014compositional}
M.~B. Horowitz, E.~M. Wolff, and R.~M. Murray, ``A compositional approach to
  stochastic optimal control with co-safe temporal logic specifications,'' in
  \emph{IEEE/RSJ International Conference on Intelligent Robots and
  Systems}.\hskip 1em plus 0.5em minus 0.4em\relax IEEE, 2014, pp. 1466--1473.

\bibitem{papusha2016automata}
I.~Papusha, J.~Fu, U.~Topcu, and R.~M. Murray, ``Automata theory meets
  approximate dynamic programming: Optimal control with temporal logic
  constraints,'' in \emph{55th IEEE Conference on Decision and Control
  (CDC)}.\hskip 1em plus 0.5em minus 0.4em\relax IEEE, 2016, pp. 434--440.

\bibitem{srinivasan2018control}
M.~Srinivasan, S.~Coogan, and M.~Egerstedt, ``Control of multi-agent systems
  with finite time control barrier certificates and temporal logic,'' in
  \emph{IEEE 57th Conference on Decision and Control (CDC)}.\hskip 1em plus
  0.5em minus 0.4em\relax IEEE, 2018, pp. 1991--1996.

\bibitem{lindemann2018control}
L.~Lindemann and D.~V. Dimarogonas, ``Control barrier functions for signal
  temporal logic tasks,'' \emph{IEEE control systems letters}, vol.~3, no.~1,
  pp. 96--101, 2018.

\bibitem{yang2019continuous}
G.~Yang, C.~Belta, and R.~Tron, ``Continuous-time signal temporal logic
  planning with control barrier functions,'' in \emph{2020 American Control
  Conference (ACC)}, 2020, pp. 4612--4618.

\bibitem{niu2020control}
L.~Niu and A.~Clark, ``Control barrier functions for abstraction-free control
  synthesis under temporal logic constraints,'' in \emph{59th IEEE Conference
  on Decision and Control}.\hskip 1em plus 0.5em minus 0.4em\relax IEEE, 2020,
  pp. 816--823.

\bibitem{liu2011false}
Y.~Liu, P.~Ning, and M.~K. Reiter, ``False data injection attacks against state
  estimation in electric power grids,'' \emph{ACM Transactions on Information
  and System Security}, vol.~14, no.~1, p.~13, 2011.

\bibitem{mo2010false}
Y.~Mo and B.~Sinopoli, ``False data injection attacks in control systems,'' in
  \emph{1st workshop on Secure Control Systems}, 2010, pp. 1--6.

\bibitem{cardenas2008research}
A.~A. C{\'a}rdenas, S.~Amin, and S.~Sastry, ``Research challenges for the
  security of control systems.'' in \emph{Proceedings of the 3rd Conference on
  Hot Topics in Security}, vol.~5.\hskip 1em plus 0.5em minus 0.4em\relax
  USENIX Association, 2008, p.~15.

\bibitem{fawzi2014secure}
H.~Fawzi, P.~Tabuada, and S.~Diggavi, ``Secure estimation and control for
  cyber-physical systems under adversarial attacks,'' \emph{IEEE Transactions
  on Automatic Control}, vol.~59, no.~6, pp. 1454--1467, 2014.

\bibitem{shoukry2017secure}
Y.~Shoukry, P.~Nuzzo, A.~Puggelli, A.~L. Sangiovanni-Vincentelli, S.~A. Seshia,
  and P.~Tabuada, ``Secure state estimation for cyber-physical systems under
  sensor attacks: A satisfiability modulo theory approach,'' \emph{IEEE
  Transactions on Automatic Control}, vol.~62, no.~10, pp. 4917--4932, 2017.

\bibitem{clark2018linear}
A.~Clark and L.~Niu, ``Linear quadratic {G}aussian control under false data
  injection attacks,'' in \emph{Annual American Control Conference
  (ACC)}.\hskip 1em plus 0.5em minus 0.4em\relax IEEE, 2018, pp. 5737--5743.

\bibitem{niu2019optimal}
L.~Niu and A.~Clark, ``Optimal secure control with linear temporal logic
  constraints,'' \emph{IEEE Transactions on Automatic Control}, vol.~65, no.~6,
  pp. 2434--2449, 2019.

\bibitem{niu2020optimal}
L.~Niu, J.~Fu, and A.~Clark, ``Optimal minimum violation control synthesis of
  cyber-physical systems under attacks,'' \emph{IEEE Transactions on Automatic
  Control}, vol.~66, no.~3, pp. 995--1008, 2020.

\bibitem{ramasubramanian2020secure}
B.~Ramasubramanian, L.~Niu, A.~Clark, L.~Bushnell, and R.~Poovendran, ``Secure
  control in partially observable environments to satisfy {LTL}
  specifications,'' \emph{IEEE Transactions on Automatic Control}, vol.~66,
  no.~12, pp. 5665--5679, 2020.

\bibitem{wolff2014optimization}
E.~M. Wolff, U.~Topcu, and R.~M. Murray, ``Optimization-based trajectory
  generation with linear temporal logic specifications,'' in
  \emph{International Conference on Robotics and Automation (ICRA)}.\hskip 1em
  plus 0.5em minus 0.4em\relax IEEE, 2014, pp. 5319--5325.

\bibitem{wang2017safety}
L.~Wang, A.~D. Ames, and M.~Egerstedt, ``Safety barrier certificates for
  collisions-free multirobot systems,'' \emph{IEEE Transactions on Robotics},
  vol.~33, no.~3, pp. 661--674, 2017.

\bibitem{ames2016control}
A.~D. Ames, X.~Xu, J.~W. Grizzle, and P.~Tabuada, ``Control barrier function
  based quadratic programs for safety critical systems,'' \emph{IEEE
  Transactions on Automatic Control}, vol.~62, no.~8, pp. 3861--3876, 2016.

\bibitem{jagtap2020formal}
P.~Jagtap, S.~Soudjani, and M.~Zamani, ``Formal synthesis of stochastic systems
  via control barrier certificates,'' \emph{IEEE Transactions on Automatic
  Control}, vol.~66, no.~7, pp. 3097--3110, 2020.

\bibitem{anand2019verification}
M.~Anand, P.~Jagtapt, and M.~Zamani, ``Verification of switched stochastic
  systems via barrier certificates,'' in \emph{IEEE 58th Conference on Decision
  and Control (CDC)}.\hskip 1em plus 0.5em minus 0.4em\relax IEEE, 2019, pp.
  4373--4378.

\bibitem{sun2020continuous}
C.~Sun and K.~G. Vamvoudakis, ``Continuous-time safe learning with temporal
  logic constraints in adversarial environments,'' in \emph{American Control
  Conference (ACC)}.\hskip 1em plus 0.5em minus 0.4em\relax IEEE, 2020, pp.
  4786--4791.

\bibitem{fujisaki1972stochastic}
M.~Fujisaki, G.~Kallianpur, and H.~Kunita, ``Stochastic differential equations
  for the non linear filtering problem,'' \emph{Osaka Journal of Mathematics},
  vol.~9, no.~1, pp. 19--40, 1972.

\bibitem{yin2011finite}
J.~Yin, S.~Khoo, Z.~Man, and X.~Yu, ``Finite-time stability and instability of
  stochastic nonlinear systems,'' \emph{Automatica}, vol.~47, no.~12, pp.
  2671--2677, 2011.

\bibitem{reif2000stochastic}
K.~Reif, S.~Gunther, E.~Yaz, and R.~Unbehauen, ``Stochastic stability of the
  continuous-time extended {K}alman filter,'' \emph{IEE Proceedings-Control
  Theory and Applications}, vol. 147, no.~1, pp. 45--52, 2000.

\bibitem{wongpiromsarn2015automata}
T.~Wongpiromsarn, U.~Topcu, and A.~Lamperski, ``Automata theory meets barrier
  certificates: Temporal logic verification of nonlinear systems,'' \emph{IEEE
  Transactions on Automatic Control}, vol.~61, no.~11, pp. 3344--3355, 2015.

\bibitem{clark2021control}
A.~Clark, ``Control barrier functions for stochastic systems,''
  \emph{Automatica}, vol. 130, p. 109688, 2021.

\bibitem{clark2020control}
A.~Clark, Z.~Li, and H.~Zhang, ``Control barrier functions for safe {CPS} under
  sensor faults and attacks,'' in \emph{59th IEEE Conference on Decision and
  Control (CDC)}.\hskip 1em plus 0.5em minus 0.4em\relax IEEE, 2020, pp.
  796--803.

\bibitem{chen2018building}
Z.~Chen, L.~Li, and X.~Huang, ``Building an autonomous lane keeping simulator
  using real-world data and end-to-end learning,'' \emph{IEEE Intelligent
  Transportation Systems Magazine}, 2018.

\end{thebibliography}

\end{document}